\documentclass{vldb}

\usepackage[show]{notes}
\usepackage[english]{babel}
\usepackage[utf8]{inputenc}
\usepackage{subfigure}
\usepackage[vlined,linesnumbered,ruled,noend]{algorithm2e}
\usepackage{amsfonts, amsmath, amssymb,wasysym,stmaryrd}
\usepackage{mathrsfs,amsthm}
\usepackage{tikz}
\usepackage{multirow}

\theoremstyle{plain}
\newtheorem{theorem}{Theorem}
\newtheorem{lemma}[theorem]{Lemma}

\newtheorem{definition}[theorem]{Definition}

\usepackage{type1cm}     % type1 computer modern font
\usepackage{graphicx}     % advanced figures
\usepackage{xspace}     % fix space in macros
\usepackage{balance}     % to better equalize the last page
\usepackage{booktabs}     % nicer tables
\usepackage[bf,tableposition=top]{caption}     % captions on top for tables
\usepackage[hyphens]{url}     % handle long urls
\usepackage{bold-extra}     % bold + {small capital, italic}
\usepackage{siunitx}          % \num for decimal grouping
\usepackage[vlined,linesnumbered,ruled,noend]{algorithm2e}     % algorithms

\usepackage[square,numbers]{natbib}     % better references
\newcommand{\customerset}{\ensuremath{\mathcal{D}}}
\newcommand{\NP}{\ensuremath{\mathbf{NP}}}
\newcommand{\bigO}{\ensuremath{\mathcal{O}}}

\newcommand{\spara}[1]{\smallskip\noindent\textbf{#1}}

\newenvironment {squishlist}
{\begin{list}{$\bullet$}
  { \setlength{\itemsep}{0pt}
     \setlength{\parsep}{3pt}
     \setlength{\topsep}{3pt}
     \setlength{\partopsep}{0pt}
     \setlength{\leftmargin}{1.5em}
     \setlength{\labelwidth}{1em}
     \setlength{\labelsep}{0.5em} } }
{\end{list}}

\begin{document}
% \title{Finding Influential and Well-Informed Users in Large-Scale Social Networks}
% \title{Scalable Facility Location on Massive Graphs with Giraph}
%\title{Scalable Facility Location on Massive Graphs}
\title{Scalable~Facility~Location for~Massive~Graphs on~Pregel-like~Systems}

\numberofauthors{1} 
\author{
\begin{tabular}{cccc}
Kiran Garimella & Gianmarco De Francisci Morales & Aristides Gionis & Mauro Sozio \\
\affaddr{Aalto University} & \affaddr{Yahoo Labs, Barcelona} & \affaddr{Aalto University} & \affaddr{Telecom Paristech}\\
\textsf{kiran.garimella@aalto.fi} & \textsf{gdfm@acm.org} & \textsf{aristides.gionis@aalto.fi} & \textsf{sozio@telecom-paristech.fr}\\
\end{tabular}
%\alignauthor Kiran Garimella\\
%	\affaddr{Aalto University}
%\alignauthor Gianmarco De~Francisci~Morales\\
%	\affaddr{Yahoo Labs, Barcelona}
%\alignauthor Aristides Gionis\\
%	\affaddr{Aalto University}
%\alignauthor Mauro Sozio\\
%	\affaddr{Telecom Paristech}
}

\maketitle

\begin{abstract}
We propose a new scalable algorithm for facility location.
Facility location is a classic problem, 
where the goal is to select a subset of facilities to open, 
from a set of candidate facilities $F$, 
in order to serve a set of clients $C$.
The objective is to minimize the total cost of opening facilities plus
the cost of serving each client from the facility it is assigned to. 
In this work, we are interested in the \emph{graph setting}, 
where the cost of serving a client from a facility is represented by the
shortest-path distance on the graph.
This setting allows to model natural problems arising in the Web and
in social-media applications.
It also allows to leverage the inherent sparsity of such graphs, as
the input is much smaller than the full pairwise distances between all
vertices.

To obtain truly scalable performance, we design a parallel algorithm that
operates on clusters of shared-nothing machines.
In particular, we target modern Pregel-like architectures, and we
implement our algorithm on Apache Giraph.
Our solution makes use of a recent result to build sketches for massive graphs, and of a fast parallel algorithm to find maximal independent sets, as building blocks.
In so doing, we show how these problems can be solved on a Pregel-like architecture, and we investigate the properties of these algorithms.
Extensive experimental results show that our algorithm scales
gracefully to graphs with billions of edges,
while obtaining values of the objective function that are competitive with a state-of-the-art sequential algorithm.
\end{abstract}

\section{Introduction}
\label{sec:introduction}

Facility location is a classic combinatorial optimization problem. 
It has been widely studied in 
\emph{operations research}~\cite{kuehn1963heuristic, manne1964plant} 
and \emph{theoretical computer science}~\cite{Shmoys1997approx, Charikar1999improved, Jain2003greedy, Thorup2001kmedian}, 
and it has been applied in many information-management tasks, 
including 
clustering data streams~\cite{guha2000clustering}, 
data compression~\cite{buchsbaum2000engineering},  
grammar inference~\cite{garofalakis2000xtract}, 
information retrieval~\cite{zuccon2012topk}, 
and
design of communication networks~\cite{mirzaian1985lagrangian,qiu2001placement}.
In the most basic setting of the problem,
we are given a set of facilities $F$,  a set of clients $C$, 
and costs 
$c(f)$ for \emph{opening} a facility $f\in F$ and 
$d(c,f)$ for \emph{serving} a client $c\in C$ with a facility $f\in F$.
The goal is to select a subset of facilities to open so that all clients
are served by an open facility and
the total cost of opening the facilities plus serving the clients is minimized. 

The facility-location problem is \NP-hard, 
but a number of different approximation algorithms are
known~\cite{Shmoys1997approx, Charikar1999improved,
  Jain2003greedy, Thorup2001kmedian}, 
Each algorithm has different characteristics, 
but they all achieve an approximation guarantee that is a small
constant, 
for example the algorithm of \citeauthor{Jain2003greedy} has an approximation
guarantee of 1.61~\cite{Jain2003greedy}.
The existing algorithms operate in the traditional
sequential model and they assume that the data resides in main memory. 

However,
facility location is a general-purpose optimization problem
that can be used in applications related to web graphs, 
large social networks, and other such massive-scale datasets, 
whose size may far exceed the memory of a single machine.
Examples of such applications include placing caches for content delivery on the Internet,
finding aggregators in information networks,
and compressing social-media activity.
Some of these application scenarios are discussed in Section~\ref{sec:scenario}.

To cope with large problem sizes, 
modern applications take advantage of large-scale distributed systems, 
such as MapReduce~\cite{Dean2004mapreduce} and Hadoop, 
or of variants targeted to graph data, 
such as Pregel~\cite{Malewicz2010pregel} 
and its open-source clones,
Giraph~\cite{Ching2011giraph} and 
Graph\-Lab~\cite{Low2012distributedgraphlab}.
Such systems offer several advantages, among which higher potential for scalability 
and a simple programming interface that requires implementing only a small number of functions.

In this paper we present the first algorithm for the facility-location
problem designed for a Pregel-like system.
In particular, we implement our algorithm on Giraph, 
thus adding facility location to the toolbox of optimization problems that
can be solved for very large datasets on modern computer clusters. 
Our work is inspired by, and follows, the large body of recent work
that offer MapReduce-type solutions to important computational problems, 
such as
linear programming~\citep{DBLP:journals/pvldb/ManshadiAGKMS13}, 
maximum cover~\citep{chierichetti2010max},
similarity join~\citep{baraglia2010document, DeFrancisciMorales2010aps},
graph matching~\citep{deFrancisciMorales2011scm,Lattanzi2011filtering, DBLP:journals/pvldb/ManshadiAGKMS13}, 
counting triangles in graphs~\citep{Suri2011triangle},
and 
finding dense subgraphs~\citep{Bahmani2012densest}. 

Most works in the the area of theoretical computer science focus on the classical formulation of the facility-location problem, where the input consists of the full $|F|\times|C|$ set of distances~\cite{Shmoys2000facilitysurvey}. Unfortunately, in this setting even algorithms with linear running time (in the size of input) are not practical when both $|F|$ and $|C|$ are large.

A more economical representation of the facility-location problem, 
which fits our applications of interest, 
assumes that $F$ and $C$ are vertices of a graph
$G=(V,E)$ (that is, $V=F\cup C$),\footnote{Note that we 
do not require $F$ and  $C$ to be disjoint. In fact, in many cases of
interest it is $V=F=C$.}
where the number of edges $m=|E|$ is very small compared to 
$|F|\times|C|$, i.e., the graph is sparse.
A client $c\in C$ can be served by a facility $f\in F$ even if $(c,f)\not\in E$, 
provided that there is a path in the graph from $c$ to $f$, 
and the cost $d(c,f)$ is the shortest-path distance metric induced by
the graph. 
We require that the running time and storage requirements of
our algorithm be quasilinear functions of $|E|$. 
If the graph $G$ is sparse, 
as most real-world graphs are, 
this leads to a significantly more scalable algorithm. 

Our approach is designed for such a graph-based formulation of the
facility-location problem.
%This design allows our method to scale to very large datasets. 
%, assuming that the application scenario fits this model.
As mentioned above, most sequential algorithms for classical facility 
location cannot cope with this setting, as they require that the full
distance matrix is provided as input, or, equivalently, that a constant-time distance oracle is available. 
A notable exception is the sequential algorithm by \citet{Thorup2001kmedian},  
which, like in our scheme, is able to take advantage of the underlying sparsity of the graph structure.

Our algorithm targets a parallel shared-nothing setting.
While other parallel algorithms have been proposed in the literature, our approach is the first one to target modern computer clusters.
In particular, \citet{Blelloch2010parallel} proposed a parallel facility-location approximation algorithm for the PRAM model.
Our work extends this parallel algorithm to the more scalable Pregel model.

There are three phases in our approach: $(i)$ \emph{neighborhood sketching}, $(ii)$ \emph{facility opening}, and $(iii)$ \emph{facility selection}.
All three phases are fully implemented in Giraph, and the code is open-source and available on GitHub.\footnote{\url{https://github.com/gvrkiran/giraph-facility-location}}

The first phase builds an all-distances sketch (ADS) that estimates the neighborhood function of each vertex.
This sketch is used in the second phase to decide when to open a facility.
To do so, we use the historic inverse probability (HIP) estimator, 
recently proposed by \citet{Cohen2014ads}.

The facility-opening phase expands balls around facilities in parallel.
It decides which facilities to open depending on the number of clients that reside within the facility-centered balls.
To estimate the number of clients inside the balls, we use the sketch created in the previous phase.

Finally, the facility-selection phase removes duplicate assignments of a client to more than one facility that might have been created due to the parallel nature of the algorithm.
To do so, it computes a maximal independent set (MIS) on the 2-hop graph of the open facilities.
For this sub-problem, we provide a parallel implementation of a recent greedy approximation algorithm~\cite{blelloch2012greedy}.

Concretely, the contributions of this paper are as follows:

\begin{squishlist}
\item 
we provide the first Pregel solution for the facility-location problem.
Our algorithm, unlike previous sequential and PRAM ones, 
is deployable on clusters available in modern computing environments;
\item
our solution uses the sparse graph-based representation of the
facility-location problem, 
which improves significantly the scalability of the method;
\item
our facility-location algorithm employs fundamental subproblems, 
such as all-distances sketch (ADS) and maximal independent set (MIS),
%on the square of a graph, 
for which we also provide the first implementations in the Pregel model;
\item
we provide an extensive experimental evaluation that shows the
scalability of our methods on very large datasets.
\end{squishlist}

%The rest of the paper is organized as follows...

\section{Application scenarios}
\label{sec:scenario}

To further motivate the use of facility location in large-scale
social-network analysis, 
we outline two application scenarios, 
both inspired by the Twitter micro-blogging platform, 
but applicable to other social-media systems as well.

\spara{Who to follow on Twitter?}
Deciding who to follow on Twitter poses a challenging trade-off: 
on the one hand we do not want to miss interesting news, 
on the other hand we want to avoid been bombarded with irrelevant information.
One way to tackle this problem is to look for a small set of ``aggregator''
users who are likely to reproduce 
(e.g., retweet, reshare, or comment upon)
most of the relevant information that we are interested in.
Additionally, we would like to receive information in a timely manner, so we also seek to minimize the delay introduced by these aggregators.

The problem of deciding who to follow on Twitter can be formulated naturally as a facility-location problem. 
%Imagine we want to make a recommendation to user $u$ regarding whom they should follow on Twitter.
The set of clients in this instance of the problem consists of the set of relevant news items for a user $u$.
The set of facilities consists of all the Twitter users that $u$ can possibly follow.

Following user $v$ corresponds to opening its facility, and it has a cost which is proportional to the amount of content that user $v$ generates (as user $u$ will have to skim over all that content).
This cost can easily be expressed in terms of time.
On the other hand, as user $v$ is likely to produce content related to certain news items, following user $v$ corresponds to serving some clients, which are the relevant news items that user $v$ will reproduce.
The cost of a news item (client) $e$ served by a user (facility) $v$ corresponds to the expected time that the item will take to reach and be reproduced by~$v$.
An estimate of the expected time required for an interesting news item to reach a certain user can be obtained by processing past logs of Twitter activity.

In summary, in this instance of the facility-location problem we are interested in selecting a set of users in order to minimize the time taken to read the content produced by these users plus the time required for interesting news items to be reproduced by these same users.

\spara{Network-based summarization of Twitter activity.}
This second application scenario addresses the problem of summarizing the overall activity in the Twitter network.
The example has a somewhat similar flavor to the first one, but it is qualitatively and quantitatively different.

In this case, we are asked to represent the whole micro-blogging activity in the network, say all the topics discussed (or hashtags), together with the identity of the users who have mentioned those topics.
A na\"{i}ve way to represent such an activity dataset is to build the set $\{(u,t)\}$, where $(u,t)$ indicates the fact that user $u$ has mentioned topic~$t$. 

Consider now {\em compressing} the activity dataset $\{(u,t)\}$ by exploiting the underlying social characteristics of the network.
In particular, as Twitter users influence their friends, it is likely that certain topics are confined to local neighborhoods of the social network.
Based on this observation, we can represent the overall activity in the network in two parts:
($i$) a seed subset of users and all the topics they have mentioned;
($ii$) every other user who has mentioned a topic can be represented by a pointer (path) to their closest seed user who has mentioned the same topic.

We can formulate this problem as an instance of the facility-location problem by using the minimum description length (MDL) principle.
Selecting seed users corresponds to opening facilities, and the cost of opening a facility corresponds to the total space
required to describe the topics discussed by the corresponding user.
Similarly, describing a topic for a non-seed user via a path to a seed user can be modeled as the service cost from a client to its
closest facility.
By using the MDL framework, both costs, opening a facility and service costs, can be neatly expressed in information bits, and the problem can be treated as a data compression problem.

This activity-summarization formulation can be seen from two different points of view:
$(i)$ as a data-compression task, per se;
and $(ii)$ as a data-clustering and data-reduction task, where the goal is to summarize the overall activity in the network by selecting the most central users, and ensuring that the rest of the activity is performed by users who are sufficiently close to the selected set of central users.

\section{Preliminaries}
\label{sec:preliminaries}

Before presenting our distributed algorithm for the facility-location problem, we briefly define the problem and specify the setting we are considering in this paper.

\subsection{Problem definition}
\label{section:problem}
In the {\em metric uncapacitated facility-location problem}, we are
given a set of facilities~$F$ and a set of clients~$C$.
For each facility $f\in F$ there is an associated cost $c(f)$ for 
{\em opening} that facility.
Additionally, a distance function 
$d: C \times F \rightarrow \mathbb{R}_{+}$ is defined
between the facilities and the clients.
The distance satisfies the triangle inequality. 
The objective in the facility-location problem is to select a set of
facilities  $S \subseteq F$ to {\em open}
in order to minimize the following objective function:
\[
\sum_{f \in S} c(f) + \sum_{c \in C} d(c,S),
\]
where $d(c,S)$ is the distance of client $c\in C$ to its 
{\em closest} opened facility, i.e.,  
\[
d(c,S)=\min_{f \in S} d(c,f).
\] 
In this paper, we are interested in the \emph{graph setting} of the
facility-location problem~\cite{Thorup2001kmedian}, 
where we are also given a weighted graph $G=(V,E, w)$, 
with $w:E \rightarrow \mathbb{R}_{+}$ a weight function on the edges of the graph. 
The sets of facilities and clients are subsets of the graph vertices $(F,C \subseteq V)$.
The distance between clients and facilities is given by the shortest-path distance on the weighted graph.

The motivation for focusing on the graph setting of the
facility-location problem is to take advantage of the fact that many
real graphs, such as web graphs and social networks, are sparse.
The goal is to leverage the sparsity of such graphs in order to
develop practical and scalable algorithms. 
Thus we aim for algorithms whose complexity is quasilinear with
respect to the size of the underlying graph.

We note that algorithms for the classical formulation of the metric uncapacitated 
facility-location problem~\citep{Blelloch2010parallel, Jain2003greedy} are not 
straightforward to adapt to the graph setting, as they require all pairwise vertex 
distances to be provided in input. 
This requirement is clearly not practical when the input graph is large.
 
In our work we consider both directed and undirected graphs. 
We develop an efficient algorithm with provable guarantees 
(on the quality of the solution delivered by the algorithm)
for the undirected case, 
while the algorithm provides a practical heuristic for the case of directed graphs.
We focus on the case when $F=C=V$, although all our claims hold for
the more general case when $F,C \subseteq V$. 

\subsection{The Giraph platform}
\label{section:giraph}

The algorithms presented in this paper are designed for the Giraph
platform~\cite{Ching2011giraph}, an Apache implementation of the Pre\-gel computational
paradigm. 
Pregel is based on the Bulk Synchronous Parallel (BSP) computation model, and can be summarized by the motto ``think like a vertex''~\citep{Malewicz2010pregel}.
At the beginning of the computation, the vertices of the graph are distributed across worker tasks running on different machines on a cluster.
Computation proceeds as a sequence of iterations called supersteps.
Algorithms are expressed in a vertex-centric fashion inside a \texttt{vertex.compute()} function, which gets called on each vertex exactly once in every superstep.
The computation involves three activities: receiving messages from the previous superstep, updating the local value of the vertex, and sending messages to other vertices.

Pregel also provides \emph{aggregators}, a mechanism for global communication and monitoring.
Each vertex can write a value to an aggregator in superstep $t$, the system combines those values via a reduction operator, and the resulting value is made available to all vertices in superstep $t + 1$.
%Pregel includes a number of predefined aggregators, such as min, max, or sum operations on various integer or string types.
Aggregators can be used for global statistics, e.g., to count the total number of edges by using a \texttt{sum} reduction.
They can also be used for global coordination.
For instance, a \texttt{min} or \texttt{max} aggregator applied to the vertex ID can be used to select a vertex to play a special role.
Additionally, one branch of \texttt{vertex.compute()} can be executed until a boolean \texttt{and} aggregator determines that all vertices satisfy some condition.

Giraph adds an optional \texttt{master.compute()} to the Pregel model. 
This function performs centralized computation, and is executed by a single master task before each superstep.
It is commonly used for performing serial computation, and for
coordination in algorithms that are composed of multiple
vertex-centric stages by using aggregators~\citep{salihoglu2014optimizing}.
Aggregators written by workers are read by the master in the following superstep, while aggregators written by the master are read by workers in the same superstep.
We employ this feature in our implementation (see Section~\ref{sec:implementation}).

\subsection{Approximate neighborhoods (ADS)}
\label{section:ADS}

The all-distances sketch (ADS) is a probabilistic data structure for 
approximating the neighborhood function of a
graph \citep{Cohen2014ads}.
Namely, ADS aims to answer the query 
\emph{``how many vertices are within distance $d$ from vertex $v$?''}.
ADS maintains a logarithmic-size sketch for each vertex.
In the sequential computational model, the total time to build the ADS is
quasilinear in the number of graph edges.
Once built, the ADS of a vertex can be used to estimate 
%% neighborhood cardinalities, i.e., 
the number of vertices within some distance.
It has also been used to estimate other properties of
the graph, such as distance distribution, effective diameter, % spid,
and vertex similarities~\citep{Boldi2011hyperanf, Cohen2013similarity}.

The ADS of a vertex $v$ consists of a random sample of vertices.  
The probability that a vertex $u$ is included in the sketch of vertex $v$ 
decreases with the distance $d(u,v)$.
The sketch contains not only the vertex $u$ but also the distance $d(u,v)$.     
The ADS can be thought as an extension of the simpler \emph{min-hash}
sketch~\citep{broder1997resemblance, cohen1997size}, 
which has been used for approximate distinct
counting~\cite{cohen1997size, durand2003loglog, Flajolet1985martin}, 
and for similarity estimation~\cite{broder1997resemblance, cohen1997size}.
The ADS of $v$ is simply the union of the min-hash sketches of all the
sets of the $\ell$ closest vertices to $v$, for each possible value of $\ell$.
Min-hash sketches have a parameter $k$ that controls the trade-off
between size and accuracy: a larger $k$ entails a better approximation
at the expense of a larger sketch. 
Essentially, $k$ controls the size of the sample.
%For an unweighted graph, the size of the ADS is bounded by $kD$, where $D$ is the diameter of the graph. 
The size of the ADS is bounded by $k\log n$.

Our algorithm relies heavily on a recently-proposed ADS structure,
the historic inverse probability (HIP) estimator \citep{Cohen2014ads}, 
which extends significantly previous variants and offers novel
estimation capabilities. 
In particular, HIP can be used to answer neighborhood queries for both
\emph{unweighted} and  \emph{weighted} graphs. 
It can also be used to answer \emph{predicated} neighborhood queries, 
that is, to approximate the number of vertices in a neighborhood that
satisfy a certain predicate on vertex attributes. 
We use this latter feature in to exclude already served
clients from the estimation of the number of clients within a ball
(see Section~\ref{sec:algorithm}).

\begin{algorithm}[t]
\caption{\label{algo:ads-sequential} Build ADS sequentially}
	\KwIn{Graph $G(V,E)$}
	\KwOut{ADS of $G$}
	
	\For{$v\in V$} {
	    ADS($v$) = $\varnothing$ \\
            BKMH($v$) = $\varnothing$  
            \tcp{Bottom-$k$ min-hash} }
	\For{$v\in V$ and $u \in \{{V}\text{ sorted by }{d(v)}\}$ }{ \tcp{
             list vertices in incr.\ distance from~$v$}
		\If{ $r(u) < \max_{r}$(BKMH($v$))}{ \tcp{$r(u)$ is the hash of $u$}
			ADS($v$) $\leftarrow$ ADS($v$) $\cup \; (u, d(v,u))$ \\
			BKMH($v$) $\leftarrow$ bottomK(BKMH($v$) $\cup \; u)$
		}
	}
	\Return ADS
\end{algorithm}

\begin{algorithm}[t]
\caption{\label{algo:ads-giraph} Build ADS in Giraph. \texttt{Vertex.Compute()}}

	\KwIn{vertex value $v$, edge values $E$, messages $M$}
	\KwOut{updated vertex value $v'$}
	\KwData{ADS = $\varnothing$; BKMH = $\varnothing$}
	\tcp{state variables are stored in the vertex $v$}

	OutMsgs = $\varnothing$ \\
	\For{ $m\in M$ }{ \tcp{the message contains the entries of the
            \\ ADS of neighbors that were updated in \\ the previous super step}
		\For { $(u,d) \in m.\text{getEntries}()$ }{
			\tcp{$u$ is the vertex.id and $d$ its distance}
			\If{ $r(u) < \max_{r}(\text{BKMH}) $ }{
				\tcp{if $u$ has already reached $v$
                                  before, \\ it will not be considered again}
				ADS $\leftarrow$ ADS $\cup \; (u, d)$ \\
				CleanUp(ADS$(u), d$)
				\tcp{for each distance, remove an entry from ADS(u) if its hash is not in the bottom-k for that distance}
				BKMH $\leftarrow$ bottomK(BKMH $\cup \; u$) \\
				OutMsgs $\leftarrow$ OutMsgs $ \cup \;
                                (u, d+e(u,v))$ \\ 
                                \tcp{for unweighted graphs \\ $e(u,v)$ is 1 \\ for weighted graphs, its the weight of the edge between $u$ and $v$}
			}
		}
	}
	\For{ $e \in E$ }{
		$\text{sendMsgTo}(e, \text{OutMsgs})$
	}
	
\end{algorithm}

Pseudocode for the sequential version of ADS is presented in
Algorithm~\ref{algo:ads-sequential}, and for the Giraph version in
Algorithm~\ref{algo:ads-giraph}.
Algorithm~\ref{algo:ads-giraph} works for both weighted and unweighted graphs.
Line 6 of Algorithm~\ref{algo:ads-giraph} performs a cleanup of the ADS, 
which removes those entries for which the hash is not in the bottom-$k$ min-hashes for a given distance.
This condition happens because the vertices are processed by ADS in Giraph as discovered by a BFS (i.e., ignoring weights), rather than sorted by their distance.
This cleanup operation can be time consuming, since it needs to sequentially access all the entries of the ADS.
For this reason, we do not perform the cleanup in each superstep, instead, we do it only periodically, when the size of the ADS becomes too large.
For unweighted graphs, the cleanup can be avoided altogether, 
since the order by which the vertices are presented to the ADS in Giraph corresponds to their distance from $v$.
Therefore the bottom-k min-hash for a given distance $d$ is always completed before the vertices at distance $d+1$ are processed.

\section{Algorithm}
\label{sec:algorithm}

As discussed earlier, our algorithm consists of three phases: 
$(i)$ \emph{neighborhood sketching}, $(ii)$ \emph{facility opening}, and $(iii)$ \emph{facility selection} via maximum independent set (MIS).
This section presents the main body of the algorithm; phases ($ii$) and ($iii$).
We first describe the PRAM version of phase ($ii$), and our Pregel version.
Then, we present our solution for the MIS problem. %%  phase ($iii$).
Finally, we describe our implementation for Giraph.
In the pseudocode presented below, \texttt{for} and \texttt{while} loops are meant to be \emph{parallel}  
(i.e., executed by all vertices in parallel), unless otherwise specified.

\subsection{PRAM algorithm for facility location}

The distributed algorithm proposed in this paper is inspired by the algorithm by \citet{Blelloch2010parallel},
which is developed for the PRAM computational model.
In our paper we adapt it to a Pregel-like platform, and we also extend
it to the graph setting, as discussed in Section~\ref{section:problem}.
For completeness, we provide a brief overview here. % by \citeauthor{Blelloch2010parallel}. 

The algorithm operates in two phases:
facility opening, and facility selection.
It starts with all facilities being \emph{unopened} and all
clients being \emph{unfrozen}.

The algorithm maintains a graph $H$ that represents
the connections between clients and open facilities. 
Initially, the vertices of $H$ consist of the set of facilities $F$
and the set of clients $C$, while its set of edges is empty, that is, $H=(F\cup C, \emptyset)$.
During the execution of the algorithm, 
if a client $c$ is to be served by a facility $f$,
the edge $(c,f)$ is added in the graph $H$.
In the first phase of the algorithm it is possible for a
client $c$ to be connected to more than one facility $f$ in $H$. 
However, the facility-selection phase 
is a ``clean up'' phase where redundant facilities are closed
so that each client is connected to exactly one facility. 

In the facility-opening phase,
each client tries to reach a facility by expanding a ball with
radius $\alpha$, in parallel. 
The expansion phase is iterative, and in each iteration the radius
of the ball grows by a factor of $(1+\epsilon)$, 
where $\epsilon$ is a parameter that provides
an accuracy-efficiency trade-off.
Initially, the radius $\alpha$ is set to a sufficiently
small value (details below).

The radius of the ball of a client $c$ is denoted by $\alpha(c)$.
If a client $c$ is unfrozen, the radius of its corresponding ball is
set to the current global value $\alpha$, 
while if a client $c$ gets frozen it does not increase the radius of
its ball anymore.
When a facility $f$ is reached by a sufficiently large number of
clients it is declared \emph{open}. This number is proportional to the cost $c(f)$ of opening that facility.
In particular, a facility $f$ is opened if the following condition is satisfied: 
\begin{equation}
\label{eq:opfac}
\sum_{c \in C}{\max\{0, (1+\epsilon)\alpha( c) - d(c, f)\}} \geq c(f).
\end{equation}

For a newly opened facility $f$, all clients $c$ within radius 
$\alpha$ from $f$ are frozen, and the edges $(c,f)$ are added to the graph $H$. 
The facility-opening phase continues by growing $\alpha$ in each iteration by a factor of $(1+\epsilon)$
as long as there is at least one unopened facility and at least one unfrozen client. 

\smallskip
In the facility-selection phase, 
if all the facilities are open but some clients are not yet frozen,
these unfrozen clients $c$ are connected to their nearest facility,
their radius is set to the distance from it, i.e., $\alpha(c) = \min_f d(c, f)$, 
and the graph $H$ is updated accordingly.

At this point, % as mentioned earlier, 
a client may be served by more than one facility in~$H$.
The final step of the algorithm consists in closing the facilities
that are not necessary, as their clients can be served by other
nearby facilities. 
This step relies on computing a maximal independent set (MIS)
in an appropriatelly-defined graph $\overline{H}$. 
In particular, $\overline{H}$ is the graph whose vertices are
the open facilities and there is an edge between two facilities
$f_a$, $f_b$ if and only if there is a client $c$ that is connected to both
$f_a$ and $f_b$ in $H$.
It is easy to see that a maximal independent set $S$ in $\overline{H}$
has the property that each client $c$ is connected to exactly one facility.
Clients whose facility is not in $S$ are assigned to the nearest open facility.
Thus, the set of open facilities $S$ returned by the algorithm is a maximal independent set on $\overline{H}$.

To complete the description of the algorithm, the initial ball radius
is set to $\alpha_{0} = \frac{\gamma}{m^2}(1+\epsilon)$, 
where $m =|F||C|$ and $\gamma$
is defined as follows. 
For each client $c \in C$  we set 
\[
\gamma_c = \min_{f \in F}  \left\{c(f) + d(c,f)\right\} ,
\]
and then 
$\gamma = \max_{c \in C} \gamma_c$. 

Pseudocode of the algorithm is shown in Algorithm~\ref{algo:classfac}.
\citet{Blelloch2010parallel} prove the
following theorem on the quality of approximation.

\begin{theorem}[\cite{Blelloch2010parallel}]
For any $\epsilon>0$, Algorithm~\ref{algo:classfac} has an
approximation guarantee of $3+\epsilon$, while the total number of
parallel iterations is 
$\bigO( \frac{1}{\epsilon} \log (|F| |C|) )$.
\label{thm:pram}
\end{theorem}

\begin{algorithm}[t]
\caption{\label{algo:classfac} PRAM algorithm for facility
  location~\cite{Blelloch2010parallel}
  }
    \KwIn{Facilities $F$, clients $C$, distance $d(\cdot,\cdot)$ \\ 
facility opening cost $c(\cdot)$, accuracy parameter $\epsilon$}
    \KwOut{Subset of opened facilities $S$}

    $O\leftarrow \emptyset$ 
        \tcp{Opened facilities}
    $U\leftarrow C$ 
        \tcp{Unfrozen clients}
    $H \leftarrow (F\cup C,\emptyset)$ 
        \tcp{Graph connecting $F$ with $C$}
    $\alpha\leftarrow\frac{\gamma}{m^2}(1+\epsilon)$ 
        \tcp{Initial ball radius}
                     
    \While{$(O \neq F) \text{ and } (U\not=\emptyset)$} {
 %     \tcp{While there are unopened facilities and unfrozen clients}  

        $\alpha \leftarrow \alpha\, (1 + \epsilon)$
           \tcp{Increase ball radius} 
           \tcp{For loops below executed in parallel} 
           \For{ $c \in U$} {
           \tcp{Set new radius for unfrozen clients}
           $\alpha(c) \leftarrow \alpha$ }
                
        \For{ $f\in D$} {
            \tcp{Open $f$ if reached by many clients}
	            If Eq.\ (\ref{eq:opfac}) is satisfied then add $f$ to $O$} %and $\alpha(f) \leftarrow \alpha$} 
        \For{ $c \in U$} {
            \If{ $\exists f\in O$ s.t.\   $(1 + \epsilon)
                  \alpha(c) \geq d(f,c)$}{
                 \tcp{if there is an opened facility nearby}
                 Remove $c$ from $U$ 
                     \tcp{Freeze client}
                 Add edge $(c,f)$ in $H$  
                     \tcp{Update $H$ graph}
            }
        }
    }
    \If {$O=F \text{ and } U\not=\emptyset$}{
        \For{$c\in U$}{
            $f^* \leftarrow \arg \min_f d(c, f)$ 
                \tcp{Nearest facility}
            Add edge $(c,f^*)$ in $H$  
        }
    }
    $\overline{H} \leftarrow (F,E)$ where $E$ contains edges $(f_a , f_b)$ if there is
          $c\in C$ such that  $(c, f_a),(c, f_b)\in E(H)$
           
    $S \leftarrow \text{MIS}(\overline{H})$ 
        \tcp{maximal independent set on $\overline{H}$} 
   \Return{ $S$}
\end{algorithm}

\subsection{Pregel-like algorithm for facility location}

We now discuss how to adapt Algorithm~\ref{algo:classfac} to the
graph setting discussed in Section~\ref{section:problem}, as well as in a Pregel-like platform, such as Apache Giraph, 
discussed in Section~\ref{section:giraph}.

The main challenges we need to tackle for adapting the algorithm are
the following:
\begin{squishlist}
\item 
leverage the sparsity of the graph $G=(V,E)$ to avoid a
quadratic blowup of distance computations between facilities and
clients;
\item 
compute efficiently, in a distributed manner, a maximal independent set on the graph $\overline{H}$.
In particular, as the graphs $H$ and $\overline{H}$ %is the square of the graph $H$, it 
may be dense, 
it is desirable to compute a MIS of $\overline{H}$ without materializing $H$ nor $\overline{H}$ explicitly.
\end{squishlist}

The latter challenge is discussed in Section~\ref{sec:mis}.
Coping with the first challenge, 
% i.e., exploiting the sparsity of the graph, 
boils down to been able to check, 
for each facility~$f$, 
whether Equation~(\ref{eq:opfac}) is satisfied, and thus deciding when to open a facility.

To this end, we rearrange the left-hand side of Equation~(\ref{eq:opfac}), 
so as to be able to evaluate it by means of the ADS algorithm discussed 
in Section~\ref{section:ADS}.
Observe that in Algorithm~\ref{algo:classfac} the $\alpha$'s take values in the range
\[ 
R = \{\alpha_{0},(1+\epsilon)\alpha_{0}, (1+\epsilon)^2\alpha_{0},\dots \}.
\]

For every facility $f$, let $N(f,d)$ be the number of clients within distance $d$ from $f$, while let $n(f,d)$ be the
number of clients whose distance from $f$ is in the range $(d/(1+\epsilon),d]$. 
Suppose that  all clients within distance $\alpha \in R$ from facility $f$ are unfrozen.
In this case, we know that for all these unfrozen clients $\alpha ( c ) = \alpha$, so we can 
rewrite the left-hand side of Equation~(\ref{eq:opfac})  as follows:
\begin{multline*}
\sum_{c \in C \mid d(c,f) \leq \alpha} {\max\{0, (1+\epsilon)\alpha( c) - d(c, f)\}} = \\ 
\sum_{{d \in R \mid d \leq \alpha}} n(f,d) \cdot \max\{0, (1+\epsilon)\alpha - d\},
\end{multline*}
where we replace $\alpha( c )$'s with $\alpha$ and rearrange the terms of the summation by grouping terms with the same value.

If some clients within distance $\alpha$ from $f$ are frozen, the former claim might not hold anymore, and we need a more sophisticated solution.
Our goal is then to maintain an approximation of the the left-hand side of Equation~(\ref{eq:opfac}) incrementally.
Let $q(f)$ denote the current approximation computed by our algorithm.
Also, for each facility $f$, let $\hat{N}(f,d)$ be the number of \emph{unfrozen} clients within distance $d$ from $f$, while let $\hat{n}(f,d)$ be the number of \emph{unfrozen} clients at distance in the range $(d/(1+\epsilon),d]$.
At each iteration of the ball-expansion phase, we add a term $t(f,\alpha)$ to $q(f)$. 
This term accounts for the increase in contribution
to $q(f)$ due to the newly-reached unfrozen clients, while subtracting
excess contribution due to previous iterations.

The increase in contribution $t(f,\alpha)$ is defined as
\begin{equation}\label{eq:contr1}
 \sum_{{d \in R \mid d \le \alpha} } \hat{n}(f,d)  \cdot  \max\{0, (1+\epsilon)\alpha-d\}   , 
\end{equation}
if $\alpha=\alpha_0$ (no excess contribution to be subtracted), and
\begin{equation}\label{eq:contr2}
 \sum_{{d \in R \mid d \le \alpha} } \hat{n}(f,d)  \cdot  (\max\{0, (1+\epsilon)\alpha-d\} -  \max \{0, \alpha -d\} ),
\end{equation}
otherwise. The term $t(f,\alpha)$ 
is added to $q(f)$ in each iteration of the algorithm for
the current value of radius $\alpha$.

The terms $\hat{N}(f,d)$ can be computed efficiently in a distributed fashion
by employing the ADS.
%while the $\hat{N}(f,d)$'s can be computed using the predicated query feature of
%the ADS scheme (see Section~\ref{sec:implementation} for further details).
Given that  $\hat{n}(f,d)= \hat{N}(f,d) - \hat{N}(f,d/(1+\epsilon))$, it follows
that also the left-hand side of Equation~(\ref{eq:opfac}) can be computed efficiently in a distributed fashion.
To show the validity of our approximation we need the following definition.

\begin{definition}
Given real numbers $a,b,\epsilon >0$, 
we say that 
$a$ approximates $b$ with accuracy $\epsilon$, 
and write
$a \approx_{\epsilon} b$,
 if $a \in [(1+\epsilon)^{-1} b , (1+\epsilon)b ]$. 
\end{definition}

\begin{lemma}
Given $\epsilon>0$, consider the quantity $q(f)$ computed as described
above, for $f \in F$.
Let $\alpha$ be the ball radius at the current step of the algorithm. The following holds:
\[ 
q(f) \approx_{\epsilon} \sum_{c \in C} 
\max\{0, (1+\epsilon)\alpha(c) - d(c, f) \}. 
\]
\label{lemma:approxeq}
\end{lemma}

\begin{proof}
The proof is by induction on the steps of the algorithm. In the first
step of the algorithm where $\alpha=\alpha_0$, 
all clients are unfrozen, so all $\alpha(c)$
are equal to $\alpha$. 
Therefore, it follows that
 \begin{eqnarray*}
q(f) &= & 
\sum_ {{d \in R \mid d \leq \alpha}} \hat{n}(f,d)  \left( (1+\epsilon)\alpha-d  \right) \\
 & \approx_{\epsilon} & \sum_{c \in C} \max\{0, (1+\epsilon)\alpha(c) - d(c, f) \}.
\end{eqnarray*} 

Now suppose the lemma holds at the $(k-1)$-th step of the algorithm. 
At step $k$, $q(f)$ must include the quantities $\max\{0, (1+\epsilon)\alpha ( c ) - d(c, f) \} $ 
for each unfrozen client $c$ at step $k$. 
If a client $c$ is unfrozen at step $k$, then the contribution of
$c$ to $q(f)$ at step $(k-1)$ is 
$\max\{0, \alpha( c ) - d(c, f) \}$ (by inductive hypothesis). 
Hence, at step $k$  it suffices to add 
\[ 
\max\{0, (1+\epsilon)\alpha( c) - d(c,f)\} - 
\max \{0, \alpha( c) - d(c,f)  \}
\]
to $q(f)$, for each unfrozen client $c$ at step $k$. 
The lemma follows from 
simple algebraic manipulations, taking into account the fact that 
$\alpha(c)=\alpha$ for each unfrozen client.
\end{proof}

\begin{algorithm}[t]
\caption{\label{algo:fastFacLoc}
Pregel-like algorithm for facility location (graph setting)}

    \KwIn{Graph $G=(V,E,d)$, facilities $F\subseteq V$, 
      clients $C\subseteq V$, facility opening cost $c(\cdot)$,
      accuracy $\epsilon$}
    \KwOut{Subset of opened facilities $S$}
    
    $O \leftarrow \emptyset$ 
        \tcp{opened facilities}
    $U\leftarrow C$ 
        \tcp{Unfrozen clients}
    $\alpha\leftarrow\alpha_0\leftarrow\frac{\gamma}{m^2}(1+\epsilon)$ 
        \tcp{Initial ball radius} 
    
    $q(f) \leftarrow 0$  for each $f \in F$
    
   \tcp{next one is a sequential while}
    \While{$(O \neq F) \text{ and } (U\not=\emptyset)$} {
        $\alpha \leftarrow \alpha\, (1 + \epsilon)$
            \tcp{Increase ball radius} 

        $\alpha( c ) \leftarrow \alpha $ for each $c \in U$ 
                	
        $\overline{O}\leftarrow$  OpenFacilities($G, F \setminus O, U, c(\cdot), \alpha, \alpha(\cdot), q(\cdot)$) \\       
        \For{$f \in \bar{O}$}{
            send($f,\alpha,$``FreezeClient'') 
            \tcp{$f$ sends a clients a "FreezeClient" message to all vertices within distance $\alpha$} 
        }
        \For{$c \in U$}{
        	   \If{$c$ receives a ``FreezeClient'' message} {$U \leftarrow U \setminus \{c\}$}
        }
        
        $O \leftarrow O \cup \overline{O}$
    }
    
        \If {$O=F \text{ and } U\not=\emptyset$}{
        \For{$c \in U$}{
            $\alpha ( c ) \leftarrow \arg \min_f d(c, f)$ 
        }
    }
           
     $S=\text{MIS}\overline{\text{H}}$($G,O,C,\alpha(\cdot)$)
     \tcp{Computing a MIS of $\overline{H}$ without building $H$ nor $\overline{H}$} 
     
    \Return{ $S$}
\end{algorithm}

\begin{algorithm}[t]
\caption{\label{algo:openfac} OpenFacilities($G, D, U, c(\cdot), \alpha, \alpha(\cdot), q(\cdot)$)}
    \KwIn{Graph $G=(V,E)$, unopened facilities $D$, 
              unfrozen clients $U$, facility opening cost $c(\cdot)$, 
              current radius $\alpha$, radius for frozen clients and opened facilities $\alpha(\cdot)$, 
              facility contribution from clients $q(\cdot)$}
    \KwOut{Newly opened facilities $\overline{O}$}

    \For{ $f\in D$ \tcp{For each unopened facility}  }  {     
        
    \tcp{use ADS Algorithm~\ref{algo:ads-giraph}} 
    Compute $\hat{n}(d,f)$ for each $f \in F$ and $d \in R$ \\
    \If{$\alpha = \alpha_0$} 
      {Compute $t(f,\alpha)$ as in Equation~(\ref{eq:contr1})}
    \Else 
      {Compute $t(f,\alpha)$ as in Equation~(\ref{eq:contr2})}
    
    % otherwise it is set as in Equation~(\ref{eq:contr2}).% \tcp{$\hat{n}(\cdot,\cdot)$ can be computed from $N(\cdot,\cdot)$ using ADS predicated queries.}
%        \[t(f,\alpha)  \leftarrow \sum_{d \in R, d\le \alpha} n(f,d)  \cdot (  \max(0, (1+\epsilon)\alpha-d)\]
%        else        
%        \begin{align*}
%        t(f,\alpha) & \leftarrow \sum_{d \in R, d\le \alpha} n(f,d)  \cdot \\ &(  \max(0, (1+\epsilon)\alpha-d) -  \max (0, \alpha -d   ) ) 
%        \end{align*}        
%        $t(f, \alpha) \leftarrow \sum_{d \in R, d\le \alpha} 
%        n(f,d) ((1+\epsilon)\alpha-d - \max\{0, \alpha'-d\})$ \\
        
        $q(f) \leftarrow q(f) + t(f,\alpha)$ \\
        \If{$q(f) \ge c(f)$} {
            add $f$ in $\overline{O}$ \\
            \tcp{$\alpha$'s allow not to materialize $H$ nor $\overline{H}$}
            $\alpha(f) \leftarrow \alpha$ 
        }
    }
    \Return{$\overline{O}$}
\end{algorithm}

Pseudocode is shown in Algorithm~\ref{algo:fastFacLoc}.
It consists of two main building blocks: 
an algorithm for deciding which facilities should be opened (Algorithm~\ref{algo:openfac}), and 
an algorithm for computing a maximal independent set of the graph $\overline{H}$ without explicitly building such a graph (Algorithm~\ref{algo:MIS}).
The pseudocode for distributing messages in the graph (denoted by the \texttt{send} procedure) is omitted for brevity. 

During the execution of the algorithm, 
for each open facility $f$ we let $\alpha(f)$ be the value of $\alpha$ when $f$ is opened.
Observe that there is an edge $(c,f)$ in $H$ only if 
(1) $\alpha(c)=\alpha(f)$, 
(2) $c$ is within distance $(1+\epsilon) \alpha(c)$ from $f$, and 
(3) $f$ is open.
Therefore, storing the values for $\alpha(f)$ and $\alpha(c)$ allows us not to materialize $\overline{H}$,
which might be very costly. 

\subsection{Maximal independent set}
\label{sec:mis}

\citet{salihoglu2014optimizing} recently proposed an implementation of the classic Luby's algorithm~\cite{luby1986simple} for computing the MIS in a Pregel-like system such as Giraph. In our approach, we need to compute a MIS of $\overline{H}$ which is essentially the graph $H^2$ after removing all unopened facilities (and their edges) from $H^2$. As we do not materialize $H$ nor $\overline{H}$, even computing the degree of a vertex in $\overline{H}$ (which is needed in Luby's algorithm) might require to exchange a large number of messages. Therefore, we resort to another algorithm developed by~\citet{blelloch2012greedy} which works as follows. 

Initially, all vertices are \emph{active} and a unique ID is assigned randomly to each of them.
This operation can be done in one parallel step by letting each of the $n$ vertices pick an integer in the range $[1,n^3]$, uniformly at random.
Then, with high probability, the vertex IDs are unique,
The ID of facility $f$ is denoted by $\pi(f)$.
Then, in every parallel step, each active vertex $v$ checks whether its ID is the minimum among its neighbors.
If this is the case, $v$ is included in the maximal independent set and all its neighbors become \emph{inactive}.
This procedure is iterated $O(\log ^2 n)$ times, after which it can be shown that the selected vertices induce a maximal independent set in the input graph, with high probability. 

As we do not materialize $H$ nor $\overline{H}$, we need to slightly modify the algorithm by~\citet{blelloch2012greedy}.
Recall that there is an edge $(c,f)$ in $E(H)$ only if 1) $\alpha(c)=\alpha(f)$, 2) $c$ is within distance $(1+\epsilon) \alpha( c)$ from $f$, and 3) $f$ is open.
Moreover, there is an edge $(f_a, f_b)$  in $E(\overline{H})$ if there exist $c \in C$ such that $(c, f_a)$ and $(c,f_b) \in E(H)$.
After determining its ID $\pi(f)$, each facility $f$ sends a message $(\pi(f),\alpha(f))$ to all vertices within distance $(1+\epsilon)\alpha(f)$ from $f$.
Each client $c$ collects all messages $(\pi(f),\alpha(f))$, and retains only the pairs $(\pi(f),\alpha(f))$ corresponding to the facilities $f$ that $c$ is connected to, i.e., $\alpha(f) = \alpha(c)$.
Then, each client computes the minimum ID $\pi_{\min}$ among all the facilities it is connected to, and sends back a message containing $\pi_{\min}$ to all such facilities.
Each facility $f$ is included in the maximal independent set if an only if $\pi_{\min} = \pi(f)$, 
in which case it sends $\pi_{\min}$ to all neighboring facilities (in $\overline{H}$) so that they are removed from the set of active vertices.
The last step is performed by letting each facility $f$ send $\pi_{\min}$ to all clients $c$ within distance $(1+\epsilon)\alpha(f)$, which in turn deliver such message to all facilities within distance $(1+\epsilon)\alpha(c)$. 
For pseudocode see Algorithm~\ref{algo:MIS}.

In Section~\ref{sec:experiments}, we evaluate the algorithms for computing a MIS proposed by \citet{blelloch2012greedy} as well as the algorithm proposed by \citet{salihoglu2014optimizing}.

\begin{algorithm}[t]

\caption{\label{algo:MIS} $\text{MIS}\overline{\text{H}}$($G,O,C,\alpha(\cdot)$) }
$S \leftarrow \emptyset, A \leftarrow O$

\For{$f \in A$}{
$\pi(f) \leftarrow RAND([1,n^3])$
}
\tcp{next one is a sequential for}
\For{$i=1,\dots, \lceil \log^2 n \rceil$}{

\For{$f \in A$}{
send($f,(1+\epsilon) \alpha(f), (\pi(f),\alpha(f))$)
}

\For{$c \in C$}{
$\pi_{\min}=\min_{ (\pi(f),\alpha(f)) :\alpha(f)=\alpha( c)} \pi(f)$
send($c,(1+\epsilon) \alpha( c), \pi_{\min})$
}

\For{$f \in A$}{
\If{$\pi_{\min} = \pi(f)$}{
	$S \leftarrow S \cup \{f\}$
	
	$A \leftarrow A \setminus \{f\}$
	
	send($f, (1+\epsilon) \alpha( f), \pi_{min}$)
}
}

\For{$c \in C$}{
	if{ $c$ receives $\pi_{\min}$,}{
	send($c,(1+\epsilon) \alpha( c),\pi_{min}$)
}

}

\For{$f \in A$}{
if{ $\pi_{\min} < \pi(f) $,}{
$A \leftarrow A \setminus \{f\}$
}
}
}
\Return $S$
\end{algorithm}

\subsection{Approximation guarantee and running time}
\label{sec:approximation}

As a consequence of Theorem~\ref{thm:pram} and Lemma~\ref{lemma:approxeq}, 
as well as from the fact that the ADS provides an approximation to the values of $\hat{n}(f,d)$, 
we are able to show a guarantee on the quality of the solution computed by our algorithm.
\begin{theorem}\label{thm:distributed}
For any $\epsilon>0$ and any integer $k \geq 1$, Algorithm~\ref{algo:fastFacLoc} has an approximation guarantee of $3+o(1) + \epsilon$. The total number of parallel iterations is 
$\bigO(\frac{D}{\epsilon} \log^2 (n))$, 
while the total number of messages exchanged by vertices is $\bigO(m)$, 
with each message requiring $\bigO(k \log n)$ bits.
\end{theorem}
The $o(1)$ term, 
which denotes a function that goes to zero as $n$ becomes large, 
is due to the ADS scheme, while $k$ is the ADS bottom-$k$ parameter.
Observe that we are able to derive the same approximation guarantees of~\citet{Thorup2001kmedian}, 
but in a distributed setting. 

Theorem~\ref{thm:distributed} holds only for undirected graphs. 
For directed graphs our guarantee does not hold, even though our algorithm can be adapted in a straightforward manner. 
In our experiments we have used directed graphs as well, 
and the performance of the algorithm is equally good.

With respect to the running time, 
the overall number of supersteps required by the algorithm is proportional to the diameter $D$ of the graph. 
This follows from the hop-by-hop communication between clients and facilities. 
Thus, 
as typically real-world graphs have small diameter, we expect our algorithm to terminate in a small number of supersteps. 

\subsection{Implementation in Giraph}
\label{sec:implementation}

The facility-opening phase consists of two main subroutines, 
which are implemented by the vertices and masters compute functions: 
\emph{ball expansion} and \emph{client freezing}.
Initially, the algorithm expands the balls around the potential facilities in parallel, 
one superstep at a time.
When one of the balls encompasses a large enough number of clients, the facility at the center of the ball opens.
At this point, the \texttt{FreezeClients} subroutine gets called to freeze all the clients within the ball.
The algorithm then resumes expanding the balls in parallel until another facility opens.
This phase terminates when no unfrozen client remains, condition monitored by the master via a \texttt{sum} aggregator.
By the end of the algorithm, vertices are either open facilities, or frozen clients with at least one facility serving them.
Clients may have multiple facilities serving them as a result of concurrent openings or intersecting balls.

We now describe in detail the implementation of this phase of the algorithm in Giraph.
The communication and coordination between the two main subroutines is particularly interesting.
In Giraph, the communication between master and workers happens via aggregators.

\spara{How to ``call a subroutine''?}
While expanding the balls, we use a boolean aggregator called \texttt{SwitchState} to monitor if any facility was opened in the current superstep.
Every vertex can write a boolean value to this aggregator, and the master can read the boolean \texttt{and} of all the values in its next superstep.
The value of the aggregator is computed efficiently in parallel via a tree-like reduction.
If \texttt{SwitchState} is true, the master writes to another aggregator \texttt{State} that represents the current function being computed.
By setting the value of \texttt{State} to \texttt{FreezeClients}, the master can communicate to the vertices to switch their computation, effectively mimicking a subroutine call.

\spara{How does \texttt{FreezeClients} work?}
The vertices execute different subroutines by switching on the \texttt{State} aggregator.
When \texttt{FreezeClients} is executed, each facility opened in the last superstep sends a ``FreezeClient'' message to all the clients within the current radius of the ball.
This message contains the ID of the facility, and the distance it needs to reach, i.e., the radius.
Each client that receives this message gets activated modifies its state to frozen by the facility whose ID is in the message, and propagates the message to its own neighbors, as explained next.
When a vertex deactivates, it writes true on the \texttt{SwitchState} aggregator.
When all the vertices terminate and deactivate, the master's \texttt{SwitchState} aggregator (which is a boolean \texttt{and}) becomes true.
The master can then resume the \texttt{OpenFacilities} routine by writing on the \texttt{State} aggregator.

\spara{How to send a message to all vertices at distance~$\mathbf{d}$?}
In Giraph, messages are usually propagated along the graph, hop-by-hop.
A vertex $v$ that wants to send a message $\mathcal{M}$ to all veritces within distance $d$, sends to each neighbor $u$ a message containing $\mathcal{M}$, as well as, the remaining distance $d-d(v,u)$, if such a distance is larger than zero.
The message is then in turn propagated by $u$ to its neighbors if the remaining distance is larger than zero.
If a vertex receives multiple copies of the same message $\mathcal{M}$ it propagates only the one with maximum remaining distance.
This subroutine takes several supersteps to complete, proportional to the distance to reach, and sends a number of messages proportional to the number of edges within distance $d$.

\spara{How to estimate the number of unfrozen clients?}
We employ ADS to estimate $N(f_i,d)$, i.e., the number of unfrozen clients within distance $d$ from $f_i$.
To do so, we use the \emph{predicated} query feature of ADS.
Given that ADS is composed by a sample of the vertices in a graph (for each possible distance), we can obtain an unbiased sample of a subset of the vertices that satisfy a predicate simply by filtering the ADS with such predicate.
That is, we can apply the condition \emph{a posteriori}, after having built the ADS.

However, there is another issue to solve in our setting.
The predicate we want to compute (unfrozen) is dynamic, as clients are frozen continuously while the algorithm is running.
Therefore, we implement this predicate by maintaining explicitly the set of frozen clients.
Whenever a client is frozen, it writes its own ID to a custom aggregator which computes the set union of all the values written in it.
At the next superstep, each facility has access to this set, and can use it to filter the ADS for the following query.
Notice that, even though this set can grow quite large, it can be approximated by using a bloom filter at the cost of decreased accuracy in the estimate.
However, our experiments are not affected by this issue, so for simplicity we do not explore the use of bloom filters, and defer its study to a later work.

\section{Experiments}
\label{sec:experiments}

We perform extensive experiments to test our approach on several
datasets by using a shared Giraph cluster containing up to 500
machines.
We design our experiments so as to answer the following questions:
\begin{squishlist}
\item[\textbf{Q1:}] What is the performance of ADS and how does it affect the main algorithm? (Section~\ref{sec:eval-ads})
\item[\textbf{Q2:}] How does our algorithm compare with
  state-of-the-art sequential ones, in terms of quality? (Section~\ref{sec:eval-facility})
\item[\textbf{Q3:}] What is the scalability of our approach in terms of time and space? (Section~\ref{sec:eval-scalability})
\item[\textbf{Q4:}] How do the two implementations of MIS compare with each other? (Section~\ref{sec:eval-mis})
\end{squishlist}

\spara{Parameters.}
There are two parameters of interest in our approach: 
the parameter $k$ of the bottom-$k$ min-hash, 
which regulates the space-accuracy tradeoff of ADS, and 
the parameter $\epsilon$, which regulates the time-accuracy tradeoff in the facility-location algorithm.

\begin{table}[t]
  \caption{Datasets.}
  \centering
  \small
  \begin{tabular}{ l  r  r l}
    \toprule
    Name & $|V|$ & $|E|$ & Description \\
    \midrule
    FF10K & 10k & 712k & \multirow{4}{*}{Forest Fire random graphs} \\
    FF100K & 100k & 11M & \\
    FF1M & 1M & 232M & \\
    FF10M & 10M & 1.6B & \\
    \midrule
    RMAT10K & $2^{13}$ & 3M & \multirow{4}{*}{R-MAT random graphs} \\
    RMAT100K & $2^{17}$ & 5M & \\
    RMAT1M & $2^{20}$ & 30M & \\
    RMAT10M & $2^{23}$ & 500M & \\
    \midrule
    ORKUT & 3M & 117M & Orkut social network\footnote{\url{http://snap.stanford.edu/data/com-Orkut.html}} \\
    TUMBLR & 10M &166M & Tumblr reblog network\footnote{The graph contains chains of re-blogs of articles on Tumblr.com. An edge $(a,b)$ indicates that user $b$ re-blogged a post by user $a$.} \\
    UK2005 & 39M & 1.4B & Web graph of the .uk domain\footnote{\url{http://law.di.unimi.it/datasets.php}} \\
    FRND & 65M & 1.8B & Friendster social network\footnote{\url{http://snap.stanford.edu/data/com-Friendster.html}} \\
    \bottomrule
  \end{tabular}
  \label{tab:datasets}
\end{table}

\spara{Datasets.}
Table~\ref{tab:datasets} summarizes the datasets used in our experiments.
We use both synthetic and real-world datasets.
% which are web graphs and social networks.
%(i) UK2005 - Web graph of the .uk domain from 2005 (39M vertices, 1.4B edges), obtained from the Laboratory of Web Algorithmics.
%(ii) Tumblr reblog network (TUMBLR) - We constructed chains of reblogs of articles on Tumblr.com. An edge $(a,b)$ indicates that user $b$ reblogged a post by user $a$. We use the largest connected component from this network (10M vertices, 166M edges).
%(iii) Friendster (FRND) - Social network data from Friendster.com (65M vertices, 1.8B edges) obtained from SNAP.
%(iv) Orkut (ORKUT) - Social network data from Orkut.com (3M vertices, 117M edges), obtained from SNAP.
We use two types of synthetic datasets and create instances with exponentially increasing sizes to test the scalability of our approach.
We choose graph-generation models that resemble real-world graphs.
The first type of synthetic graphs is generated by using the Forest Fire (FF) model~\citep{leskovec2007graph} with the following parameter values:
the forward burning probability is set equal to $0.3$ and the backward equal
to $0.4$.
The second type of synthetic graphs uses the recursive matrix model
(RMAT)~\citep{chakrabarti2004r} with parameters $a = 0.45$, $b =
0.15$, $c = 0.15$, and $d = 0.25$.\footnote{For both models, FF and
  RMAT, we use the default parameters that the data generators come
  with.}
This model can only generate graphs with a number of vertices that is a power of $2$.
%Properties of some of the random graphs we use commonly in our experiments are given in Table~\ref{tab:datasets}. 
For weighted graphs, we assign weights between~$1$ and~$100$, uniformly at
random.

\begin{figure}[t!]
\centering
\includegraphics[width=0.4\textwidth]{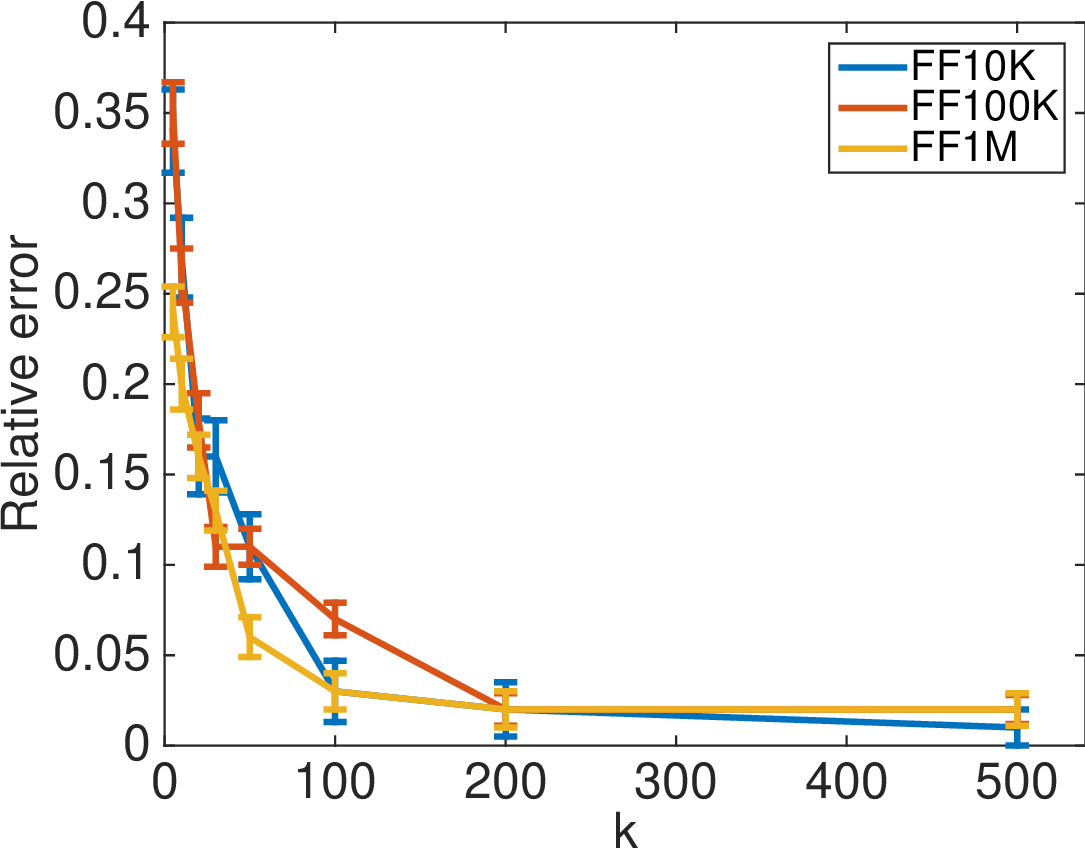}
\includegraphics[width=0.4\textwidth]{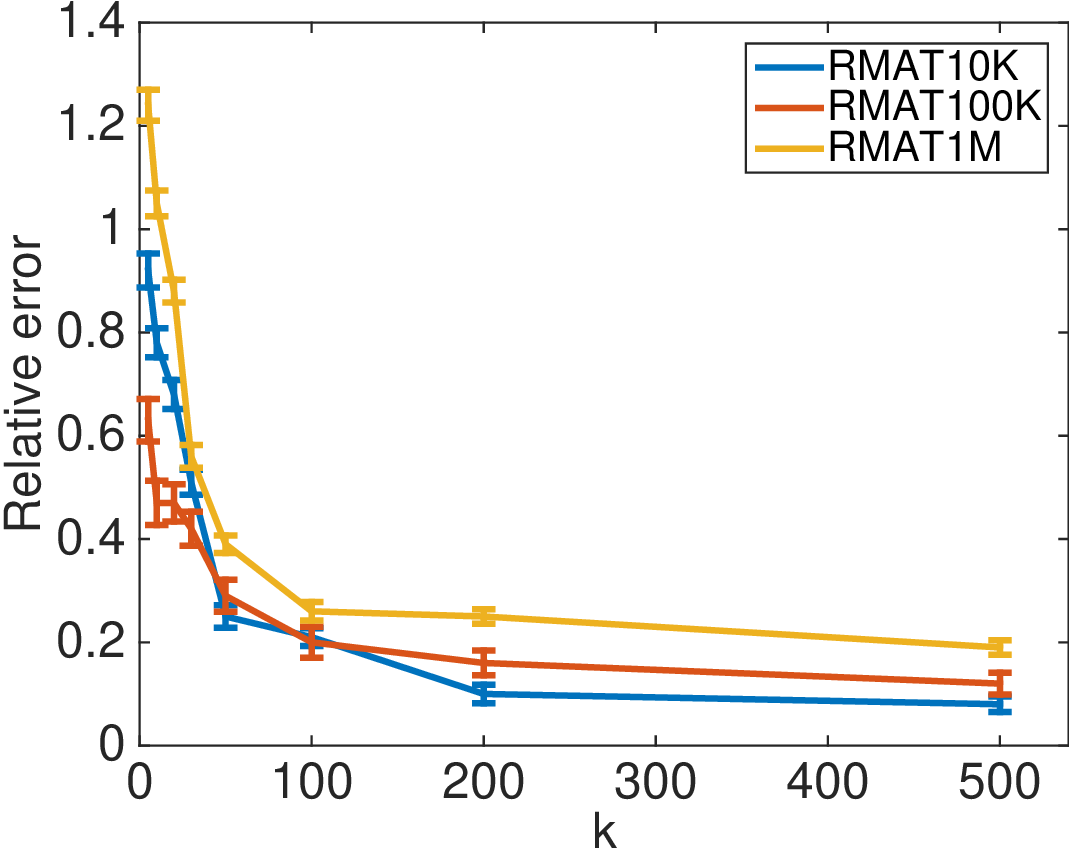}
\includegraphics[width=0.4\textwidth]{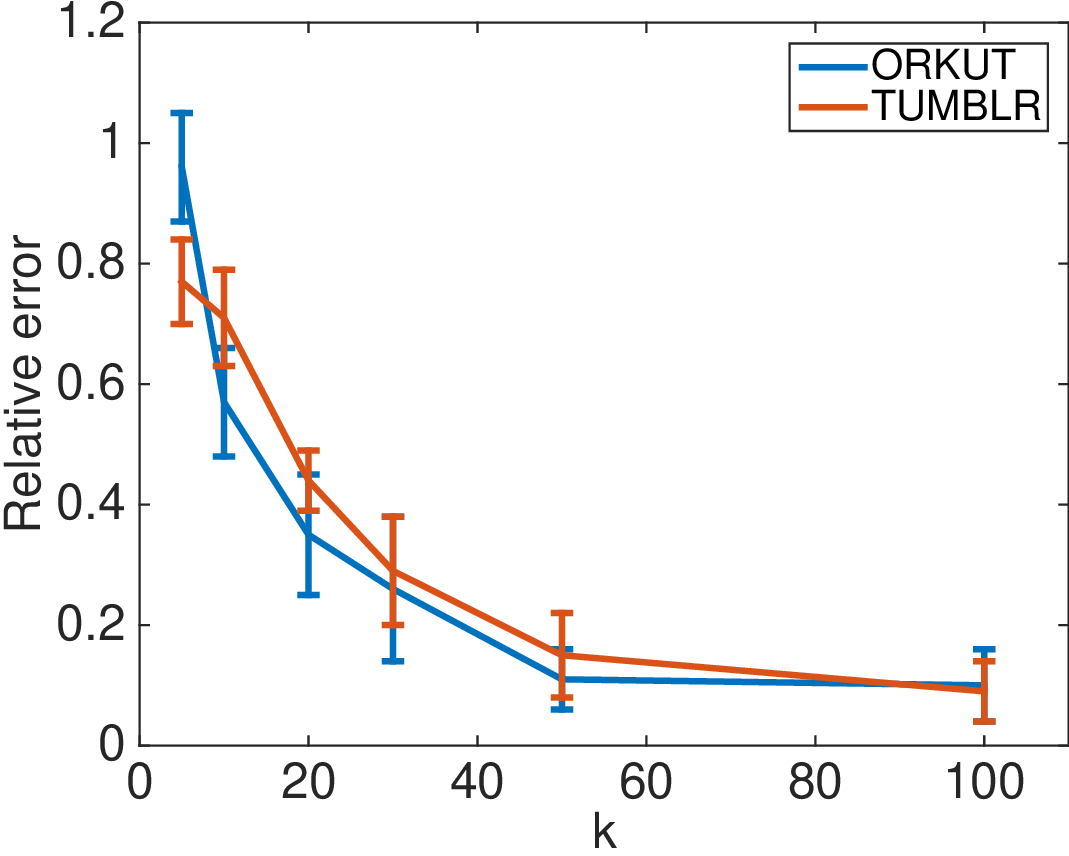}
\caption{\label{fig:ads_quality_unweighted}ADS relative error vs.\ $k$ (unweighted graphs).}
\end{figure}

\subsection{Evaluation of ADS}
\label{sec:eval-ads}

As described in Section~\ref{section:ADS}, we approximate neighborhood sizes by using the ADS data structure. %``all distance sketches''
To the best of our knowledge, we are the first to implement and test ADS on Giraph on large scale. 
In this section, we perform experiments to validate our choice and
assess the quality of the ADS estimates.
First, we evaluate the quality of ADS approximation by comparing
against exact neighborhood sizes.
Then, we experiment with the time taken for computing the ADS as a function
of~$k$. 

\spara{Accuracy vs.\ $\mathbf{k}$}:
To evaluate the accuracy of the estimates produced by ADS, we need to
compute exact neighborhood sizes. 
Since such computation is infeasible for large graphs, we compute
exact neighborhood sizes on a sample of vertices. 
For each neighborhood distance
(from 1 to 20 for unweighted graphs, 
and from 100 to 2000, at increments of 100, for weighted graphs), 
we sample 100 random vertices and compute their exact neighborhood sizes. 
%% by computing the shortest paths from the sampled vertex to all other vertices.
For each sampled vertex and each distance, 
we compute the relative error as 
$|S_E - S_\mathrm{ADS}|/S_E$, 
where $S_E$ is the exact neighborhood size and 
$S_\mathrm{ADS}$ the ADS estimate. 
Relative error averages and variances across the 2000 samples are reported in
Figure~\ref{fig:ads_quality_unweighted} (unweighted) and Figure~\ref{fig:ads_quality_weighted} (weighted).

We can see that the estimates are of high quality. 
In most cases, 
even for small values of $k$,
the average relative error is less than~$50\%$.
The variance is also small.

Note that in
Figures~\ref{fig:ads_quality_unweighted}
and~\ref{fig:ads_quality_weighted} we do not report accuracy for the largest of our datasets.
The reason is that computation of exact neighborhood sizes becomes a bottleneck.

\begin{figure}[t]
\centering
\includegraphics[width=0.4\textwidth]{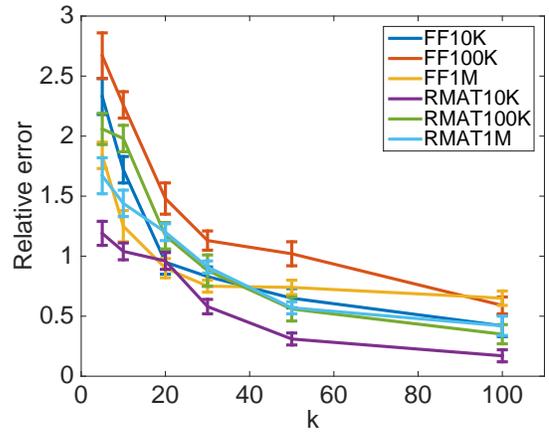}
\caption{\label{fig:ads_quality_weighted}ADS relative error vs.\ $k$ (weighted graphs).}
\end{figure}

\begin{figure}[t]
\centering
\includegraphics[width=0.4\textwidth]{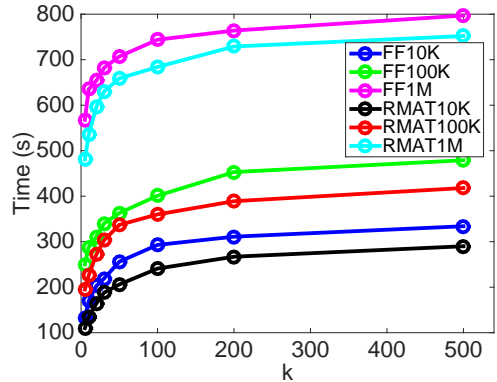}
\caption{ADS time taken vs.\ $k$ (unweighted graphs).}
\label{fig:ads_timetaken}
\end{figure}

\spara{Time vs.\  $\mathbf{k}$}:
Next, we measure the time taken to compute the ADS as a function of $k$, as shown in Figure~\ref{fig:ads_timetaken}. 
We clearly see that even for graphs with 1 million vertices and $k$ as
large as 500,  the algorithm finishes in less than 800 seconds.

\spara{Space requirements}: 
Since increasing the value of $k$ does not increase the time taken by
the algoritm, one would assume that we could use a very high value of
$k$ in order to improve the quality of approximation of ADS. 
The bottleneck, though, is the size of the ADS, which is proportional
to~$k$.
Thus, increasing $k$ increases the memory requirements of ADS.
In our setting, we have a limitation of 3.5 GB of memory per machine, 
which makes it infeasible to store an ADS for large graphs with very large
values of $k$ 
(say, $n > $ 10m and $k > $ 200).
Recall, however, that the size of the ADS is proportional to $n k\log n$, 
thus, the value of $k$ can increase linearly with the number of available machines.

\begin{table}
  \caption{\label{tab:sequential}Relative cost of the Giraph algorithm against the sequential one ($k=200$).}
  \centering
  \begin{tabular}{l r r r r r}
    \toprule
    Type & $|V|$ & $|E|$ & $\epsilon=0.01$ & $\epsilon=0.1$ & $\epsilon=1$ \\
    \midrule
    FF & 1k & 11k & 1.21 & 1.46 & 2.56 \\
    FF & 2k & 25k & 1.15 & 1.60 & 2.45 \\
    FF & 3k & 60k & 1.07 & 1.75 & 2.47 \\
    FF & 4k & 67k & 1.08 & 1.48 & 2.16 \\
    FF & 5k & 121k & 1.05 & 1.5 & 2.13 \\
    FF & 6k & 206k & 1.01 & 1.41 & 2.02 \\
    FF & 7k & 268k & 1.03 & 1.33 & 1.67 \\
    FF & 8k & 380k & 1.03 & 1.25 & 1.55 \\
    FF & 9k & 520k & 1.01 & 1.18 & 1.41 \\
    FF & 10k & 712k & 1.02 & 1.09 & 1.43 \\
    RMAT & $2^{10}$ &100k & 1.08 & 1.55 & 1.88 \\
    RMAT & $2^{11}$ & 200k & 1.06 & 1.36 & 1.7 \\
    RMAT & $2^{12}$ & 500k & 1.05 & 1.35 & 1.73 \\
    RMAT & $2^{13}$ & 800k & 1.05 & 1.24 & 1.44 \\
    RMAT & $2^{14}$ & 1000k & 1.02 & 1.14 & 1.39 \\
    \bottomrule
  \end{tabular}
\end{table}

\subsection{Facility-location algorithm}
\label{sec:eval-facility}

In this section, we evaluate the quality of the Giraph implementation of the algorithm presented in Section~\ref{sec:algorithm}.
We first compare our algorithm against a simple sequential baseline,
in terms of the cost function, for different values of the accuracy
parameter~$\epsilon$. 
We then evaluate the performance and the running time of the algorithm
as a function of~$\epsilon$.

\spara{Comparison with sequential algorithm}:
We use the sequential approximation algorithm by~\citet{Charikar1999improved} as a baseline.
The algorithm is a simple local-search method that achieves an
approximation ratio of ($2.414+\epsilon$) and has running time of
$\tilde{\bigO}(n^2/\epsilon$). 

Note that the sequential algorithm assumes the availability of all-pairs shortest path distances, which is computationally very
expensive, even for small graphs.
Therefore, we perform our evaluation with graphs consisting of no more than $10\,000$ vertices.

Table~\ref{tab:sequential} shows the results of the comparison in
terms of {\em relative cost}, which is defined as the cost of the
sequential algorithm divided by the cost of our algorithm, 
for different values of~$\epsilon$.
A smaller value means that our algorithm is competitive with the baseline.
We can see that, for small graphs, our algorithm performs quite well, 
even for large values of~$\epsilon$.

\spara{Cost vs.\ accuracy ({$\epsilon$})}:
Table~\ref{tab:sequential} shows the relative cost of our algorithm
(compared again to the sequential algorithm) with respect to the
accuracy. 
As expected, we get better solutions for smaller values of $\epsilon$,
but the solution does not get much worse even for large values of~$\epsilon$.

\spara{Running time vs.\ accuracy ({$\epsilon$})}:
We also measure the time taken by our algorithm as a function
of~$\epsilon$. 
Figure~\ref{fig:timetaken_eps} shows these results. 
From the figure we see that, as expected, our algorithm scales
linearly with respect to the size of the graph,\footnote{A linear fit
  gives $R^2$ values between $0.8$ and $0.9$.}
and is faster for larger values of~$\epsilon$.

\begin{figure}
\centering
\includegraphics[width=0.4\textwidth]{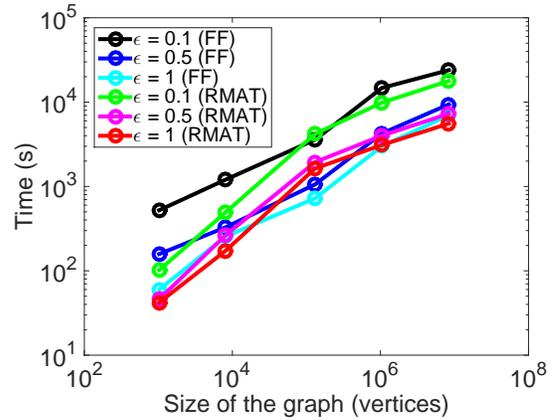}
\caption{Time taken by the algorithm for different values of~$\epsilon$ on several graphs.}
% (top: FF, bottom: RMAT).}
\label{fig:timetaken_eps}
\end{figure}

\subsection{Scalability}
\label{sec:eval-scalability}

Next, we examine the scalability of the different phases of our algorithm. 
% We experiment with data of increasing size,  for both real-world and synthetic graphs.
Recall that the three phases of our algorithm are 
($i$) ADS computation (pre-processing),
($ii$) facility-location algorithm, and 
($iii$) MIS computation (post-processing). 
Figure~\ref{fig:timetaken_all_phases} presents the time taken, 
for different datasets, broken down by phase.
The total running time, for various graph sizes, is also shown in Figure~\ref{fig:timetaken_total}.

% Figure~\ref{fig:timetaken_all_phases} replaces Tables~\ref{tab:time_forestfire},~\ref{tab:time_rmat} and~\ref{tab:time_real}.

For the results shown in Figure~\ref{fig:timetaken_all_phases}, 
we have used 
%% graph type is unweighted, 
$k = 20$, $\epsilon = 0.1$, 
and $200$ machines. 
Since the running time on a distributed cluster depends on various
factors, such as the current load of the machines, we repeat all experiments
three times and report the median running time.

\begin{figure}
\centering
\includegraphics[width=0.4\textwidth]{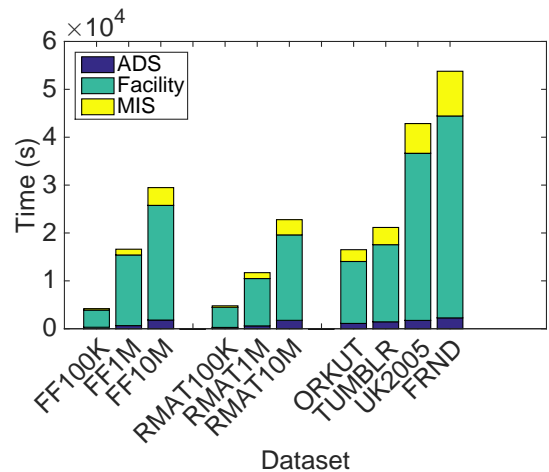}
\caption{Time taken by each phase of the algorithm.}
\label{fig:timetaken_all_phases}
\end{figure}

\begin{figure}[t]
\centering
\includegraphics[width=0.4\textwidth, clip=true, trim=0 0 0 0]{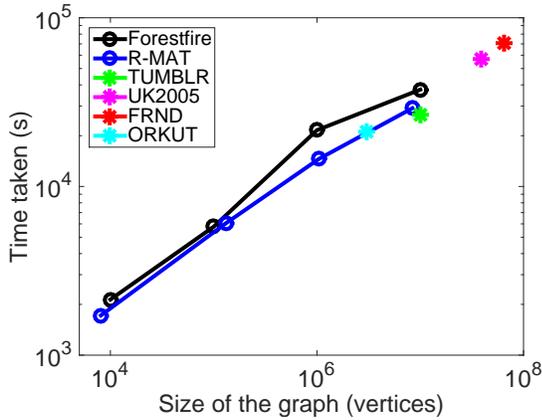}
\caption{Total time taken by the algorithm.}
\label{fig:timetaken_total}
\end{figure}

\subsection{Luby's vs.\ parallel MIS}
\label{sec:eval-mis}

As discussed above, 
we implement two methods for finding the maximal independent
set (MIS): Luby's classic one~\cite{luby1986simple}, 
which was also implemented recently by~\citet{salihoglu2014optimizing},
and a recent algorithm by~\citet{blelloch2012greedy}.

%A giraph implementation for Luby's algorithm was presented by
%\citet{salihoglu2014optimizing}, and we have extended it for the
%``square'' of a graph.
%
%We also present a first giraph implementation of the algorithm by
%\citet{blelloch2012greedy}

We compare the two methods in terms of
total time taken and number of supersteps needed to converge. 
Table~\ref{tab:lubys_mis} shows the results. 
We see that our parallel MIS algorithm is at least $3$ to $5$ times
faster than Luby's algorithm.

\begin{table}
  \caption{Comparison of two implementations of MIS in Giraph, in terms of supersteps and time taken (median over three runs).}
\centering
  \begin{tabular}{ l r r r r }
    \toprule
    Graph & \multicolumn{2}{c}{Supersteps} & \multicolumn{2}{c}{Time (s)} \\
    \cmidrule(lr){2-3} \cmidrule(lr){4-5}
    & Luby's & MIS & Luby's & MIS \\
    \midrule
    FF10K & \num{750} & \num{29} & \num{730} & \num{104} \\
    FF100K & \num{3473} & \num{85} & \num{1869} & \num{296} \\
    FF1M & \num{6119} & \num{325} & \num{6155} & \num{1205} \\
    FF10M & \num{20154} & \num{1613} & \num{11744} & \num{3711} \\
    \midrule
    RMAT10K & \num{645} & \num{17} & \num{616} & \num{87} \\
    RMAT100K & \num{3200} & \num{73} & \num{1576} & \num{319} \\
    RMAT1M & \num{5832} & \num{285} & \num{4109} & \num{1232} \\
    RMAT10M & \num{17557} & \num{1533} & \num{9492} & \num{3181} \\
    \bottomrule
  \end{tabular}
  \label{tab:lubys_mis}
\end{table}

\smallskip

\section{Related work}
\label{section:related}

%\note{
%discuss related work with respect to
%\begin{itemize}
%\item
%algorithms for facility location, sequential, distributed
%\item
%map-reduce, pregel computation, and giraph in general
%\item 
%map-reduce and giraph for graph problems
%(triangle counting, matching, PageRank, connected components, etc.)
%\item 
%other applications of facility location
%\end{itemize}
%}

\spara{Facility location.}
Facility location is a classic optimization problem.
%% with many variants \citep{Drezner1995survey}. 
\iffalse
problem with a vast literature in computer science.
It was introduced in the '60s by \citet{Hakimi1964facility}, who also proved its connection to the k-median problem.
Several variants of the problem exist, e.g., some use capacitated facilities that can serve only a finite number of customers, other consider ``obnoxious'' facilities and try to maximize the distance from customers.
For a survey see the work by \citet{Drezner1995survey}.
\fi
The traditional formulation (metric uncapacitated facility location) is \NP-hard, and so are many of its variants.
Existing algorithms rely on techniques such as LP rounding, local search, primal dual, and greedy.
The greedy heuristic obtains a solution with an approximation guarantee of $(1+\log|\customerset|)$~\citep{Hochbaum1982median}, 
%% by reduction from the set cover problem~\citep{Hochbaum1982median}.
while constant-factor approximation algorithms have also been introduced~\citep{Shmoys1997approx,Charikar1999improved}.
The approximation algorithm with the best factor so far ($1.488$) is very close to the approximability lower bound ($1.463$)~\citep{Li2011facloc1488}.
%% while the most efficient constant-factor approximation algorithm for the classical formulation of the problem is a $3$-approximation algorithm running in $O(n^2)$ where $n$ is the total number of facilities and clients~\cite{DBLP:conf/focs/MettuP00}.

%% \note{To double check that the fastest is a 3-approx algorithm.}

%% See  \citet{Shmoys2000facilitysurvey} and \citet{Vygen2004notes} for a survey.

Differently from most previous work, the input to our algorithm is a sparse graph representing potential facilities and clients and their distances, rather than the full bipartite graph of distances between facilities and clients.
Note that building the full bipartite graph requires computing all-pairs of distances and implies an $\mathcal{O}(n^2)$ algorithm.
\citet{Thorup2001kmedian} considers a setting similar to ours, and provides a fast sequential algorithm $\tilde{\mathcal{O}}(n+m)$.
%% \note[gdfm]{I think the \~O notation means a $\log n$ factor added}
%However, the algorithm is ``exceedingly complicated, and hence unlike to be of practical relevance'' (author's excerpt).

\citet{Blelloch2010parallel} propose a parallel approximation algorithm for facility location in the PRAM model.
In this work, we extend the former algorithm to work in a more realistic shared-nothing Pregel-like model.
Other parallel algorithms have also been proposed~\citep{Gehweiler2006dist,Moscibroda2005facility,Pandit2009return}.
%\todo{How is our algorithm different from any of the above? Mauro?}

\spara{Applications.}
Facility location is a flexible model 
%% to think about allocation problems that has emerged in the context of logistics~\citep{Melo2009review}.
%% However, it 
that
has been applied successfully in many domains, such as city planning,
%%  (bus stops, hospitals, warehouses), 
telecommunications, 
%%  (cell towers, switching centers), 
electronics, 
%% (circuit wiring, chip manufacturing), 
and others. 
For an overview of applications, please refer to the textbook of \citet{Hamacher2002facility}.
Furthermore, applications of facility location in social-network analysis are provided in our motivation scenarios in Section~\ref{sec:scenario}.

%\todo[gdfm]{Should it be much longer? The book I cite contains lots of applications.}
%
%\todo[aris]{I think that the sentence above is OK. The applications in the book are not sore relevant to our story.}
% 
%\todo[gdfm]{Should it have a paragraph on applications to social networks? (e.g., rumor spreading, influence maximization, centrality)}
%
%\todo[aris]{It would be good, but do we have specific references? If it is to us to invent the applications, we have Section \ref{sec:scenario}.}

\spara{Large-scale graph processing.}
MapReduce~\cite{Dean2004mapreduce} is one of the most popular paradigms used for mining massive datasets.
Many algorithms have been proposed for various graph problems, such as
counting triangles~\citep{Suri2011triangle},
matching~\citep{deFrancisciMorales2011scm,Lattanzi2011filtering}, and 
finding densest subgraphs~\citep{Bahmani2012densest}.

However, given the iterative nature of most graph algorithms, 
MapReduce is often not the most efficient solution. 
Pregel~\citep{Malewicz2010pregel} is large-scale graph processing platform that supports a vertex-centric programming paradigm and uses the bulk synchronous parallel (BSP) model of computation.
Giraph~\citep{Ching2011giraph} is an open-source clone of Pregel.
It is the platform that we use in this work.
Other distributed systems for graph processing have recently been
proposed, for instance, Signal/Collect~\citep{Stutz2010signalcollect}, GraphLab/PowerGraph~\citep{Low2012distributedgraphlab,Gonzalez2012powergraph}, GPS~\citep{Salihoglu2013gps}, and GraphX~\citep{Gonzalez2014graphx}.
Most of the APIs of these system follow the gather-apply-scatter (GAS)
paradigm,  which can be readily used to  express our algorithm.
However, the BSP model is still used due to its simplicity and ease of use.

\spara{Algorithms.}
Our work takes advantage of a number of successful algorithmic
techniques. 
We use the all-distance-sketches (ADS) and the historic inverse probability
(HIP) estimator by \citet{Cohen2014ads} to estimate the number of
vertices within certain distance from a given vertex. 
HIP is a cardinality estimator similar to HyperLogLog counters~\citep{Flajolet2007hyperloglog} and Flajolet-Martin counters~\citep{Flajolet1985martin}.
%% \todo{make sure we speak about the algorithm before?}
HyperANF~\citep{Boldi2011hyperanf} is a related algorithm that approximates the global neighborhood function of the graph by using HyperLogLog counters, but it is not directly usable in our case as we need separate neighborhood functions for each vertex.

\section{Conclusions}
We have shown how to tackle the facility-location problem at scale by using Pregel-like systems.
In particular, we addressed the graph setting of the problem, which allows to represent the input in sparse format as a graph.
We leveraged graph sparsity to tackle problem instances whose size is much larger than previously possible.

Our algorithm is composed by three phases: $(i)$ \emph{neighborhood sketching}, $(ii)$ \emph{facility opening}, and $(iii)$ \emph{facility selection}.
We implemented all three phases in Giraph, and published the code as open-source software.
For the first phase, we showed how to use ADS with HIP, a recent graph-sketching technique.
We adapted an existing PRAM algorithm with approximation guarantees for the second phase.
Finally, for the third phase we proposed a new Giraph algorithm for the maximal independent set (MIS), which is much faster than the previous state-of-the-art.
Our approach was able to scale to graphs with millions of vertices and billions of edges, thus adding facility location to the toolset of algorithms available for large-scale problems.

This work opens up several new research questions.
%The algorithm we use for phase $(ii)$ has approximation guarantees when a distance oracle is available, however, these guarantees do not necessarily hold when used in combination with a sketch such as ADS.
%An interesting open question is to prove such approximation guarantees.
From the point of view of the practitioner, this algorithm enables to solve large-scale facility-location problems, thus is a candidate for real-world applications in Web and social-network analysis.
A more general question is whether better algorithms exist for the setting we consider.
% (graph setting on a cluster).
Also, it would be interesting to know whether there are any primitives % (akin to aggregators) 
that the system could offer to develop better algorithms.

\bibliographystyle{nourlabbrvnat}
\bibliography{biblio-abbrv}

\end{document}